\documentclass[12pt,a4paper]{article}

\usepackage{macros}

\usepackage[normalem]{ulem}

\title{One Code Fits All: Strong stuck-at codes for versatile memory encoding}
\author{
    Roni Con\thanks{Technion - Israel Institute of Technology,
                    Haifa, Israel,
                    roni.con93@gmail.com}
    \and
    Ryan Gabrys\thanks{University of California San Diego,
                    Naval Information Warfare Center, CA,
                    rgabrys@ucsd.edu}
    \and
    Eitan Yaakobi\thanks{Technion - Israel Institute of Technology,
                    Haifa, Israel,
                    yaakobi@cs.technion.ac.il}
}
\begin{document}

\maketitle
\begin{abstract}
In this work we consider a generalization of the well-studied problem of coding for ``stuck-at'' errors, which we refer to as ``strong stuck-at'' codes. In the traditional framework of stuck-at codes, the task involves encoding a message into a one-dimensional binary vector. However, a certain number of the bits in this vector are 'frozen', meaning they are fixed at a predetermined value and cannot be altered by the encoder. The decoder, aware of the proportion of frozen bits but not their specific positions, is responsible for deciphering the intended message.
    We consider a more challenging version of this problem where the decoder does not know also the fraction of frozen bits. 
    We construct explicit and efficient encoding and decoding algorithms that get arbitrarily close to capacity in this scenario. 
    Furthermore, to the best of our knowledge, our construction is the first, fully explicit construction of stuck-at codes that approach capacity.

\end{abstract}
\newpage

\section{Introduction}
In this research, we initiate the development of \emph{strong-stuck-at codes}, an advanced version of traditional codes that have applications to \emph{stuck-at memories}. Our approach considers a storage medium analogous to a one-dimensional vector with a fixed length, containing a certain proportion of `frozen' components that cannot be altered during encoding. The objective is to create a coding system capable of encoding the greatest possible amount of information while ensuring the frozen components' values and positions, known during encoding but unknown during decoding, remain intact. Previous studies typically assume knowledge of the maximum size of the set of frozen components at the time of encoding and decoding, even if the set itself is not known. Our study addresses the more flexible (yet challenging) scenario where both the specific set and the maximum size of the frozen components are unknown at the decoding stage.

The problem of constructing codes for stuck-at memories has its roots in the early work of Kuznetsov and Tsybakov \cite{kuznetsov1974coding}. Building on this, Tsybakov expanded the scope by considering scenarios where, apart from the frozen components, the memory might incur additional errors post-encoding \cite{tsybakov1975defect}. This led Heegard to innovate a new class of codes, termed partitioned linear block codes \cite{heegard1983partitioned}, which he demonstrated to meet the Shannon capacity in specific conditions \cite{heegard1985capacity}. However, these findings are not applicable to scenarios involving binary alphabet codes, which is the primary focus of our study. It's noteworthy that this issue has evolved with the advent of newer technologies like Flash and Phase-Change Memory (PCM) and some new works on this (and similar) settings include \cite{lastras2010algorithms,kim2013coding,wachter2015codes,mahdavifar2015explicit}.

A strongly related area of work is in the setting of coding for ``Write-Once-Memories'' or WOM, which was originally introduced by Rivest and Shamir in 1982 \cite{rivest1982reuse}. In this setting, memory cells are initialized to each have value $0$ and, at each round of the encoding, one is allowed to change some fraction of the cells only from $0$ to $1$. For the case of two-write WOM-codes, in the first round the encoder is permitted to change any fraction of the cells to $1$. The decoding in the first round is straightforward. In the second round, the encoder has access to the state of the memory after the first round so that it knows which cells were set to one in the first round, but the decoder only has access to the state of the memory after the second write and so it does not know what bits were set to 1 in the first round. Thus, the second round of encoding/decoding represents an instance of the defective memory with stuck-at components.

Capacity-achieving two-write WOM-codes have been known for some time starting with the seminal work by Sphilka \cite{shpilka2013new} and later by Chee et al. \cite{chee2019explicit}. 
In fact, using this connection between two-write WOM-codes and coding for stuck at errors, it was noted in \cite{shpilka2013new} that if the encoder is allowed to transmit a small amount of side information directly to the decoder that cannot be corrupted by stuck-at errors, then a slight variation of the encoder/decoder for his two-write WOM is equivalent to a stuck-at code. In the work by Chee et al. \cite{chee2019explicit}, which leverages spreads in projective geometry in order to guide the encoding function for the second round write, the value of the cells matters so that it is not clear how to make their approach account for frozen or stuck-at cells that can have value $0$ or $1$. 

Perhaps the closest existing work to the problem of designing strong-stuck-at codes is the work Gabizon and Shaltiel \cite{gabizon2012invertible}, who designed capacity-achieving stuck-at codes for the case where the maximum number of frozen components is known ahead of time. Although their constructions provided the first explicit scheme with asymptotically optimal rate, their model permitted a randomized encoding function which was allowed to succeed with randomized polynomial time (with respect to the block length of the memory).

In this work, we develop almost capacity-achieving strong-stuck-at codes where the number of frozen components is not known beforehand. Although our primary goal is the design of explicit and efficient codes for this generalized model, our codes also have several properties for the classical stuck-at model. Unlike the work of Gabizon and Shaltiel, our encoding procedure is completely deterministic (see \Cref{thm:deterministic-const}). Furthermore, we show that in the randomized version of our algorithm which is presented in \Cref{sec:const-w-assumption}, we are able to construct codes using fewer random bits than in previous constructions. 

The rest of this paper is organized as follows. In the remainder of this section, we formally introduce our problem setup and highlight our results. In Section 2, we present an existential result showing that, perhaps surprisingly, it is possible to encode at virtually the same rate as a conventional stuck-at code even when the size (or a bound on the size) of the set of frozen components is not available to the decoder. Section 3 presents a simplified version of our construction where we assume the encoder is provided a side channel to convey a small amount of information to the decoder in a manner analogous to the setting originally studied in \cite{shpilka2013new}. Finally Section 4 presents our main construction.

\subsection{Problem setup}
Denote $\subsetSize{i}{N}:= \{ \cF \subseteq [N] \mid |\cF| = i
 \}$, i.e., all the subsets of $[N]$ of size $i$.
Formally, our goal is to design
\begin{enumerate}
    \item A sequence of pairs $E := (E_1, \cM_1), (E_2, \cM_2), \ldots, (E_N, \cM_N)$ where the $\cM_i$'s are sets of messages and
    \[
    E_i:\zo^{N} \times \subsetSize{i}{N} \times \cM_i \rightarrow \zo^{N}
    \]
    are encoding maps that get as input a cover vector $\bfv \in \zo^N$, a set of frozen indices of size $i$, and a message to encode.
    \item A decoder 
    \[
    D: \zo^{N} \rightarrow \bigcup_{i=1}^N \cM_i
    \]
    that maps vectors to messages.
\end{enumerate}


A \emph{strong-stuck-at-code} of \emph{length} $N$ is a pair $(E, D)$ such that for every $i\in [N]$, $\bfv\in \zo^N$, $\cF\in \bP^{i}([N])$, and $\bfm \in \cM_i$, the following two conditions hold:
\begin{enumerate}
    \item Consistency:
    \[
    (E_i(\bfv, \cF, \bfm))_j = \bfv_j\;, \quad \forall j\in \cF \;.
    \]
    Namely, the encoders are allowed to change only coordinates of $\bfv$ whose indices are outside of $\cF$.
    \item Unique-decodability:
    \[
    D(E_i(\bfv,\cF, \bfm)) = \bfm \;.
    \]
\end{enumerate}

\begin{defi}
    The rate of a strong-stuck-at-code at $\rho$-fraction defect is defined as 
    \[
    \frac{\log\left(\left|\cM_{\rho N}\right|\right)}{N} \;.
    \]
\end{defi}
    
    Naturally, given that $\rho N$ of the bits are frozen, we can encode up to $1-\rho$ fraction of information bits. Our goal in this paper is to design codes that approach this bound. This goal motivates the following code definition.

\begin{defi}
    Let $\varepsilon > 0$. An $\varepsilon$-gapped strong-stuck-at-code of length $N$ is a strong stuck-at code such that for every defect fraction $\rho \in (0, 1- \varepsilon)$, the rate of the code is at least $1 - \rho - \varepsilon$. 
\end{defi}

\subsection{Our results}

In the following theorem, we show that there are $\varepsilon$-gapped strong-stuck-at-code.
\begin{restatable}{thm}{existencialConst} \label{thm:existencial-const}
    For every $\varepsilon > 0$, there exists an $N(\varepsilon)$ such that for every $N > N(\varepsilon)$, there exists an $\varepsilon$-gapped strong-stuck-at-code of length $N$. 
\end{restatable}

Our next theorem presents a randomized construction of $\varepsilon$-gapped strong-stuck-at-code.
\begin{restatable}{thm}{randomizedConst} \label{thm:randomized-const}
    For every $\varepsilon > 0$, there exists an $N(\varepsilon)$ such that for every $N > N(\varepsilon)$, there exists a randomized $\varepsilon$-gapped strong-stuck-at-code of length $N$ such that
    \begin{enumerate}
        \item The encoder and the decoder run in $\cO{N\cdot \poly{\log N}\cdot \poly{1/\varepsilon}}$.
        \item The number of random bits that are used by the encoder is $\cO{\frac{1}{\varepsilon}\log N}$ and the encoder succeeds with probability $1 - o(1)$.
    \end{enumerate} 
\end{restatable}

Our next theorem is a version of \Cref{thm:randomized-const} that is fully deterministic. We note that the cost of making the encoder deterministic results in much higher encoding complexity.
\begin{restatable}{thm}{deterministicConst} \label{thm:deterministic-const}
    For every $\varepsilon > 0$, there exists an $N(\varepsilon)$ such that for every $N > N(\varepsilon)$, there exists an explicit $\varepsilon$-gapped strong-stuck-at-code of length $N$ such that the encoder runs in time $N^{\cO{1/ \varepsilon}}$ and the decoder runs in time $\cO{N\cdot \poly{\log N}\cdot \poly{1/\varepsilon}}$.
\end{restatable}

\subsection{Preliminaries}
For an integer $k$, we denote $[k]:= \{1,2,\ldots, k\}$. We shall denote vectors by boldface letters such as $\bfu$ and sets by calligraphic letters such as $\cF$. Note that we denote an interval of positive integers by $[a, b]$ and a vector restricted to a set of coordinates will be denoted by $\bfv_{\cF}$. In particular, a subvector of $\bfv$ starting from index $a$ up to an index $b$ will be denoted as $\bfv_{[a,b]}$.
We shall denote a concatenation of two vectors, $\bfu$ and $\bfv$ by $\bfu \circ \bfv$.
Throughout this paper, $\log x$ will refer to the base-$2$ logarithm.

We note here that in many places we drop all floors and ceilings in order to ease notation and the analysis of the codes. However, the loss in the rate due to these roundings is
negligible and does not affect the asymptotic results.

A concept that will be useful in our construction is that of \emph{almost} $k$-wise independent random variables. 
\begin{defi}
    A random variable $X = (X_1, X_2, \ldots, X_r) \in \zo^r$ is said to be $\mu$-almost $k$-wise independent if for all sets of $k$ distinct indices $\{i_1, \ldots, i_k\} \subseteq[r]$ and for all $(x_1,x_2,\ldots, x_k)\in \zo^r$, we have
    \[
    \left| \Pr[X_{i_1} = x_1, \ldots, X_{i_k}=x_k]  - 2^{-k}\right| \leq \mu \;.
    \]
\end{defi}

The following well-known result gives an efficient construction of a collection of $\mu$-almost $k$-wise independent random variables which can be generated from a small number of random bits.
\begin{lemma}[\cite{alon1992simple}] \label{lem:k-wise}
    For every two positive integers $r, k$ and every $\mu >0$, there exists a function $g:\zo^t \rightarrow \zo^r$ with $t=\cO{\log \left( \frac{k\log r}{\mu}\right)}$, such that $g(U_t)$ is a $\mu$-almost $k$-wise independent variable over $\zo^r$, where $U_t$ denotes the uniform distribution over $\zo^t$.
    Moreover, $g(\bfu)$ can be computed in time $\poly{r, 1/\mu}$. 
\end{lemma}

\begin{remark} \label{rem:g-run-time-and-random-bits}
    We shall use \Cref{lem:k-wise} with $r = \cO{N \log N}$, $k = \cO{\log N}$, and $\mu = N^{-\cO{1}}$. In this case, we have that $t = \cO{\log N}$. Furthermore, it can be verified that in this case, the running time of $g$ on an input $\bfu\in \zo^t$ is $\cO{N \cdot \poly{\log N}}$. The details are given in the appendix.
\end{remark}
We have the following simple claim whose proof is deferred to the appendix.
\begin{claim} \label{clm:full-rank-pro}
    Let $m<n$ be positive integers and 
    \[
    A = 
    \begin{pmatrix}
    A_{1,1} & A_{1,2} & \cdots & A_{1,n} \\
    A_{2,1} & A_{2,2} & \cdots & A_{2,n} \\
    \vdots  & \vdots & \ddots & \vdots \\
    A_{m,1} & A_{m,2} & \cdots & A_{m,n} \\
    \end{pmatrix}
    \]
    where $(A_{i,j})_{1\leq i\leq m, 1\leq j \leq n}$ is a $\mu$-almost $n$-wise independent variable. Then, the probability that $A$ does not have full rank is most $2^{m-n} + \mu 2^m$.  
\end{claim}

\section{Existential result}
In this section, we prove \Cref{thm:existencial-const} which is restated for convenience
\existencialConst*

\begin{proof}
    Let $L$ be an integer such that $L \leq 2\varepsilon^{-1} \leq L+1$ and let $N$ be an integer such that $L+1$ divides $N$. 
    For every $i\in [L]$, define $\cM_i:= \left\lbrace (i, \bfm) \mid \bfm\in \zo^{\frac{N}{L+1}\cdot i} \right \rbrace$.
    Every $\cM_i,i\in [L]$ can be seen as a message space of a specific length, and our encoder, based on the fraction of frozen symbols, will encode a message from the largest possible message space.

    Our strategy will be to randomly assign vectors from 
    $\zo^{N}$ into $\left|\cup_{i=1}^L \cM_i\right|$ bins 
    where each bin will be labeled 
    $B_{i,\bfm}$. 
    Formally, every $\bfv\in \zo^{N}$,
    \[
        \Pr\left[ \bfv \text{ is assigned to } B_{i,\bfm} \right] = \frac{1}{L\cdot 2^{\frac{N}{L+1}\cdot i}} \;.
    \]
    
     

    Our encoder, which receives as input a vector $\bfv\in \zo^N$, a set $\cF\subseteq [N]$ of size $\rho N$ performs the following:
    \begin{enumerate}
        \item Sets $j$ to be the largest integer such that $(1-\rho)N \geq \frac{j}{L+1} N + \frac{\varepsilon}{2} N$.
        \item Encodes a message $\bfm \in \zo^{\frac{j}{L+1} N}$ by choosing a vector $\bfu\in B_{j,\bfm}$ such that $\bfv_{\cF}= \bfu_{\cF}$ and will store this vector in the memory. 
    \end{enumerate}
    Note that by the choice of $j$, we have ensured that the gap between the length of the message and the number of unfrozen bits is at least $\varepsilon/2 \cdot N$.  
    Clearly, the decoder who knows the partition of $\zo^N$ to the sets $B_{i,\bfm}$ will correctly identify the message. 
    Thus, it remains to show that the consistency condition holds with high probability. Namely, that with high probability the second step of our encoder always succeeds. 
    
    We first compute the probability for a specific $\cF$ of size $\rho N$ and a cover vector $\bfv$, there is no $\bfu\in B_{j,\bfm}$ for which $\bfu_{\cF} = \bfv_{\cF}$. Since there are $2^{N - |\cF|}$ vectors $\bfu \in \zo^{N}$ such that $\bfu_{\cF} = \bfv_{\cF}$, the probability that none of them falls in $B_{j,\bfm}$ is at most
    \[
     \left(  1 - \frac{1}{L\cdot 2^{\frac{N}{L+1}\cdot j}} \right) ^{2^{N-|\cF|}}
    \leq \left( 1 - \frac{1}{L\cdot 2^{\frac{N}{L+1}\cdot j}}\right)^{2^{\frac{j}{L+1}N + \frac{\varepsilon}{2} N}} 
    \leq \exp\left( - \frac{1}{L}\right)^{2^{\frac{\varepsilon}{2} N}}
    \leq \exp\left( - \varepsilon \right)^{2^{\frac{\varepsilon}{2} N}} \;.
    \]
     Now, the probability that there exists a vector $\bfv\in \zo^N$, a set $\cF\subseteq [N]$ and a message $\bfm$ (of suitable length) such that the respective set $B_{j,\bfm}$ does not contain a vector that agrees with $\bfv$ on the coordinates specified by $\cF$ is at most,
     \[
     2^N \cdot 2^N \cdot 2^N \cdot \exp\left( - \varepsilon \right)^{2^{\frac{\varepsilon}{2} N}} = \exp \left( \ln2^{3N} - \varepsilon 2^{\frac{\varepsilon}{2} N} \right) \;.
     \]
    Thus, since $\varepsilon$ is constant, the probability that our partition of $\zo^N$ to the sets $B_{i,\bfm}$ indeed yields a strong-stuck-at-code is at least $1 - o(1)$ (the term $o(1)$ goes to zero as $N$ tends to infinity). 
    For every $\rho$, the rate of our probabilistic construction at $\rho$-fraction of defect is at least $1 - \rho - \frac{1}{L+1} - \frac{\varepsilon}{2} \leq 1 - \rho - \varepsilon$. 
\end{proof}

\section{Construction with clean transmission assumption} \label{sec:const-w-assumption}
In this section, we will assume that the encoder can transmit $\cO{\frac{1}{\varepsilon} \log N}$ bits to the decoder where this transmission is errorless. The decoder will use this clean metadata to decode the original message. This construction is a first step towards our final construction which does not assume that there is a clean transmission of bits between the encoder and the decoder.
Throughout this section, we assume that $C$ is a universal constant (independent of $N$) that is known both to the encoder and the decoder. 
Also, we denote by $\Bin{d}{b}$ the function that takes as input an integer $d\in [0, b - 1]$ and outputs its binary representation using $\ceil{\log b}$ bits. 

Our encoding algorithm is given in Algorithm~\ref{alg:encode-with-assumption} and the decoding algorithm is given in Algorithm~\ref{alg:decode-with-assumption}.
In the rest of the section, we prove that 

\begin{thm} \label{thm:randomized-w-assumption-const}
    Let $\varepsilon > 0$. there exists an $N(\varepsilon)$ such that for every $N > N(\varepsilon)$, there exists a randomized $\varepsilon$-gapped strong-stuck-at-code of length $N$ such that
    \begin{enumerate}
        \item The encoder uses $\cO{\frac{1}{\varepsilon}\log N}$ random bits and succeeds with probability $1 - \cO{1/\log N}$ 
        \item The encoder can transmit $\cO{\frac{1}{\varepsilon}\log N}$ bits to the decoder in an errorless transmission.
        \item The encoder and the decoder run in time $\cO{N \cdot \poly{\log N}\cdot \poly{1/\varepsilon}}$ 
    \end{enumerate}
\end{thm}

\paragraph{Comparison with \cite[Theorem 7.1]{shpilka2013new}}
    Note that although our primary aim is to design efficiently strong-stuck-at codes, our work represents an improvement over the setup previously studied by Shpilka \cite[Theorem 7.1]{shpilka2013new} where we assume we have access to a small area of clean memory (equivalently, we have an errorless transmission between the encoder and the decoder) and also the decoder knows the number of stuck-at bits. The next theorem more precisely states the previous work by Shpilka, which will be useful as a basis for comparison.
     
    \begin{thm} \cite[Theorem 7.1]{shpilka2013new}
        Let $\rho < 1$ and let $\bfv\in \zo^N$ containing $\rho N$ frozen bits. There is a randomized encoder and a deterministic decoder such that
        \begin{enumerate}
            \item The encoder can encode $(1 - p - \varepsilon)N$ bits for any constant $\varepsilon > 0$.
            \item The encoder transmits $\cO{\log ^3 N}$ bits to the decoder using an errorless transmission.
            \item The encoder runs in polynomial time in $N$ and $1/\varepsilon$.
        \end{enumerate}
    \end{thm}

    Note that the construction presented in this section requires only $\cO{\varepsilon^{-1} \cdot \log (N)}$ bits to be transmitted to the decoder in the errorless transmission compared to the $\cO{\log ^3 N}$ bits required by \cite{shpilka2013new}. \footnote{We note that in fact, we could have defined $B = C \cdot \log(N/ \log N)$. In that case, the number of random bits is $\cO{\frac{1}{\varepsilon} \cdot \log (N/ \log N)}$ at the expanse of failure probability which increases to $1 - O(1/C)$. We chose to present the first version for the sake of notations.}

We present also a deterministic version of \Cref{thm:randomized-w-assumption-const}
\begin{thm} \label{thm:explicit-w-assumption-const}
    Let $\varepsilon > 0$. there exists an $N(\varepsilon)$ such that for every $N > N(\varepsilon)$, there exists a explicit $\varepsilon$-gapped strong-stuck-at-code of length $N$ such that
    \begin{enumerate}
        \item The encoder can transmit $\cO{\frac{1}{\varepsilon}\log N}$ bits to the decoder in an errorless transmission.
        \item The encoder runs in time $N^{\cO{1/\varepsilon}}$ and the decoder runs $\cO{N\cdot \poly{\log N}}$.
    \end{enumerate}
\end{thm}

\paragraph{Notations and preliminaries for Algorithm~\ref{alg:encode-with-assumption}}
    The following notations are used in Algorithm~\ref{alg:encode-with-assumption}
    \begin{itemize}
        \item Let $B = C\cdot \log N$.
        \item We divide $[N]$ into $M:=N/B$ contiguous blocks.
        \item Let $\cF_i\subseteq \cF$ denote the frozen elements that appear in the $i$th block and $\overline{\cF}_i$ the nonfrozen elements in the $i$th block.
        \item Denote by $\rho_i = |F_i|/B$ the fraction of frozen symbols in $i$th block.
    \end{itemize}
    We proceed with a high-level description of Algorithm~\ref{alg:encode-with-assumption}, which consists of three steps where each of which is described in the next three paragraphs.
    
    
    Our encoding algorithm will encode the message into $M$ blocks, each of length $B$. At Step~\ref{enc1:step1}, for each block $i$, we compute $m_i$, the number of message bits we will encode in the $i$th block. Note that some of the $m_i$s can be zero as it can be the case that (almost) all the bits of a block are frozen. The total number of bits that are going to be encoded in the $i$th block is denoted by $\overline{m_i}$ and will contain $m_i$, another $\log N $ bits for the position of the next block to be decoded, and another $\log B$ bits that denote the number of encoded message bits in the next block. If we cannot encode message bits in the $i$th block (this happens if we have at most $2\log N + \log B$ unfrozen bits in the block), then we set $\overline{m_i} = 0$. 

    At Step~\ref{enc1:step2}, we generate $B\cdot N$ bits that are $\varepsilon$-almost $B$-wise independent. The first $\overline{m_1}B$ bits will form the matrix $A_1$, then the next $\overline{m_2}B$ bits will form the matrix $A_2$, etc. Overall, at the end of this step, we have $M$ matrices $A_1, \ldots, A_M$. 

    At step~\ref{enc1:step3}, we perform the actual encoding. We only encode bits of our message in blocks for which $\overline{m_i}\neq 0$. For each such block, we solve the linear system 
    \[
    (A_i)_{\overline{\cF}_i} \cdot \bfw_i = \bfm_i \circ \Bin{i'}{\log N} \circ \Bin{m_{i'}}{\log B}
    \]
    where $i'$ is the next block index for which $\overline{m_{i'}}\neq 0$. 
    We note that this step might fail since it can be that $(A_i)_{\overline{\cF}_i}$ does not have full rank. We will prove that this happens with small probability. 
    Finally, we concatenate all the blocks to produce our encoded cover object. Also, we transmit to the decoder the metadata that he needs to decode the message (recall that we assume that this transmission is errorless). This metadata includes the string that generates the matrices $A_1, \ldots, A_M$, the position of the first block that encodes message bits, and the number of message bits that are encoded in that block.

\begin{figure}
	\begin{algorithm}[H] 
    	\DontPrintSemicolon
        \LinesNumberedHidden
		\SetNlSty{textbf}{[}{]}
			
		\SetKwInOut{Input}{input}
		\SetKwInOut{Output}{output}
        \SetKwInOut{Notations}{notations}

        \Input{A vector $\bfv \in \zo^N$, a set of frozen indices $\cF\subseteq [N], |\cF| =\rho N $, and message $\bfm\in \zo^m$ where $m \leq N\left(1 - \rho - \frac{2}{C} - \frac{\log(C\log N)}{C\log N}\right)$}
		\Output{A vector $\bfw \in \zo^N$ and $\bfu \in \zo^{*}$.}

        \nlset{1} \label{enc1:step1} 
        \For{every $i\in [M]$}{
            \uIf{$B(1 - \rho_i) > 2\log N + \log B$}{
                Set $m_i := \min \left( B(1- \rho_i) - 2\log N - \log B, m\right)$ \;
                Set $\overline{m_i} := m_i + \log N + \log B$ \;
                Update $m = m - m_i$\;
            }
            \uElse{
                Set $\overline{m_i} := 0$ \;
            }
        }
    
        \nlset{2} \label{enc1:step2} 
        Let $r = B\cdot N$,  $\mu = N^{-C}$, $k = B$ and let $t$ be as given in \Cref{lem:k-wise}. Sample $\bfu_t$ uniformly at random from $\zo^t$ and apply the function $g$ (given in \Cref{lem:k-wise}) to get $\bfa\in \zo^{B\cdot N}$. Use the first $B\cdot m$ bits of $\bfa$ to construct $M$ matrices  
            \[A_1\in \zo^{\overline{m_1}\times B}, \ldots, A_M\in \zo^{\overline{m_M} \times B}\]
        
        \nlset{3} \label{enc1:step3}  
        Let $1\leq i_1< \cdots <i_{M'}\leq M$ be all the indices for which $\overline{m_{i_{j}}} \neq 0$. Also let $i_{M' + 1} = 0$ and $m_{i_{M'+1}} = 0$\;
        
        \For{every $j\in [M']$}{
            \If{$(A_{i_j})_{\overline{\cF}_{i_j}}$ is not full rank}{Declare failure and exit}
            Compute $\bfw_{i_j}\in \zo^{B}$ such that 
            \begin{enumerate}
                \item $A_{i_j}\cdot \bfw_{i_j} = \bfm_{i_j} \circ \Bin{i_{j+1}}{N/ \log N} \circ \Bin{m_{i_{j+1}}}{B}$
                \item $(\bfw_{i_j})_{F_{i_j}} = (\bfv_{i_j})_{F_{i_j}}$
            \end{enumerate}
        }
        \nlset{4} \label{enc1:step4} Return the string $\bfw = \bfw_1\circ \bfw_2 \circ \cdots \circ \bfw_M$ and the string $ \bfu = \bfu_t \circ \Bin{i_1}{N/ \log N} \circ \Bin{m_1}{B}$\;
        
        \caption{Encoding with assumption}
        \label{alg:encode-with-assumption}
        
    \end{algorithm}
\end{figure}

\begin{figure}
	\begin{algorithm}[H] 
    	\DontPrintSemicolon
        \LinesNumberedHidden
		\SetNlSty{textbf}{[}{]}
			
		\SetKwInOut{Input}{input}
		\SetKwInOut{Output}{output}
        \SetKwInOut{Notations}{notations}

        \Input{A vector $\bfv \in \zo^N$ and $\bfu \in \zo^{t + \log B}$}
		\Output{A message $\bfm\in \zo^*$}

        \nlset{1} Identify from $\bfu$ the vector $\bfu_t$, and the values $i$ and $m_i$\;
        \nlset{2} Compute $g(\bfu_t)$ to get a string $\bfa \in \zo^{B\cdot N}$\;
        \nlset{3} 
        \While{$i \neq 0$}{
            Identify the matrix $A_i\in \zo^{m_i\times B}$ from the string $\bfa$\;
            Compute $A_i \bfv_i$ to get $\bfm_i$ and update the next $i$ and $m_{i}$\;
        }
        \nlset{4} Return $\bfm = \bfm_1 \circ \cdots \circ \bfm_M$ \;
        \caption{Decoding with assumption}
        \label{alg:decode-with-assumption}
    \end{algorithm}
\end{figure}

\subsection{Analysis}
\paragraph{Rate.} The number of message bits that we encode in each block $i$, is $m_i = \max( B(1 - \rho_i) - 2\log N - \log B, 0)$.
Thus, the number of bits that we can encode is at least
\begin{align*}
    \sum_{i=1}^M m_i &\geq \sum_{i=1}^M B(1 - \rho_i) - 2\log N - \log B\\
    &= N - |\cF| - 2M\log N -M\log B
\end{align*}
and therefore, the rate of the scheme is at least 
\begin{equation} \label{eq:const1-rate}
    1 - \rho - \frac{2}{C} - \frac{\log (C\log N)}{C\log N} \geq 1- \rho - \frac{3}{C} \;,
\end{equation}
where the inequality follows for large enough $N$.

The following proposition proves the correctness of the algorithm.

\begin{prop} \label{prop:randomized-w-assumption-const}
    Let $\bfv\in \zo^N$, $\cF\subseteq [N]$ where $|\cF| = \rho N$. Let $\bfm \in \zo^{m}$ where $m \leq N(1 - \rho - \frac{3}{C})$ be a message to be encoded. 
    If we execute Algorithm~\ref{alg:encode-with-assumption} on $\bfv$, $\cF$, and $\bfm$, then the following holds
    \begin{enumerate}
        \item The algorithm succeeds with probability at least $1 - \cO{1/(C\log N)}$. Specifically, the only step that might cause the algorithm to fail and abort is Step~\ref{enc1:step3}.
        \item If the algorithm succeeds and outputs the vector $\bfw$ and the metadata $\bfu$, then the decoding algorithm, Algorithm~\ref{alg:decode-with-assumption}, which receives $\bfw$ and $\bfu$ as input, will output $\bfm$.
    \end{enumerate}  
\end{prop}
\begin{proof}
    Recall that we denote by $\cF_i\subseteq \cF$ the frozen elements that appear in the $i$th block and by $\overline{\cF}_i$ the nonfrozen elements that are in the $i$th block of $\bfv$.
    The step that might fail in Algorithm~\ref{alg:encode-with-assumption} and cause the algorithm to abort is Step~\ref{enc1:step3}. If one of the matrices $(A_1)_{\overline{\cF}_1}, \ldots, (A_M)_{\overline{\cF}_M}$ does not have full rank, say $(A_1)_{\overline{\cF}_1}$, then clearly we have that $\{ (A_1)_{\overline{\cF}_1} \cdot \bfw \mid \bfw \in \zo^{B-\rho_1 B}\}\subsetneq \zo^{\overline{m_1}}$. 
    Thus, there exists a vector in $\zo^{\overline{m_1}}$ that cannot be encoded using this procedure.
    Therefore, in order to be able to encode any message, we must require that the matrices $(A_i)_{\overline{\cF}_i}$ each have full rank.

    We compute the probability that $(A_1)_{\overline{\cF}_1}$ does not have full rank. Note that $(A_1)_{\overline{\cF}_1} \in \zo^{\overline{m_1}\times (B - \rho_1 B)}$ where $\overline{m_1} \leq B - \rho_1 B - \log N$. Also, it is easy to see that a random variable that is $\mu$-almost $k$-wise independent, is also $(2^{k-k'}\mu)$-almost $k'$-independent for every $k'<k$. 
    Thus, according to \Cref{clm:full-rank-pro}, the probability that $(A_1)_{\overline{\cF}_1}$ does not have full rank is at most
    \[
    2^{- \log N} + \mu\cdot 2^{B \rho_1}\cdot  2^{B(1 - \rho_1) - \log N} = \frac{2}{N}\;.
    \]
    Now by union bound, the probability that there exists a matrix among the matrices $(A_1)_{\overline{\cF}_1}, (A_2)_{\overline{\cF}_2},\ldots, (A_M)_{\overline{\cF}_M}$ that does not have full rank is at most 
    $M \cdot (2/N) = 2/(C \log N) = \cO{1/(C \log N)}$. Therefore, Step~\ref{enc1:step3} can fail with probability at most $1- \cO{1/(C\log N)}$. 
    If all the matrices are indeed full rank, then the linear equations at Step~\ref{enc1:step3} all have solutions and therefore, the encoded vector $\bfw\in \zo^N$ is just the concatenation of all the $\bfw_i$s. 
    The second output of the encoder is a vector $\bfu$ which concatenates the string $\bfu_t$ generated in Step~\ref{enc1:step2} with the position of the minimal $i\in [M]$ for which $\overline{m_i} \neq 0$ and the corresponding $m_i$. The last two values correspond to the position of the first block that encodes message bits and the number of message bits that are encoded in this block, respectively.
    
    Note that the decoder, which has access to $\bfu$ and knows the value $C$ can read the first $t$ bits to identify $\bfu_t$. Then, reading the following $\log N + \log B$ bits, the decoder knows the identity of the first block, $i\in [M]$, that encodes message bits and the exact number of bits $m_i$ the block encodes. 
    Note here that $i\in [M]$ where $M = N/ (C\log N) $ and that $m_i \leq B - 2 \log N - \log B$, therefore, $\log N$ followed by $\log B$ bits suffice in order to save $i$ and $m_i$, respectively.
    
    Computing $g(\bfu_t)$, the decoder identifies the matrix $A_{i}$ and then simply computes $A_i \bfw_i$ to get $\bfm_i$ and the position of the next block that contains information, $i<j$, and the number of encoded bits in $\bfw_j$. The decoder continues until he reaches the last block containing information. Note that this process stops. Indeed, when the decoder decodes the last block, he encounters that $0$ is encoded as the position of the next block that contains information.    
\end{proof}

We are now ready to prove \Cref{thm:randomized-w-assumption-const}.
\begin{proof}[Proof of \Cref{thm:randomized-w-assumption-const}] 
    Let $C$ be a universal constant and define $\varepsilon = 3/C$. Let $N$ be a large enough integer (depending only on $\varepsilon$) and let $\bfv\in \zo^{N}$ be a cover vector, and $\cF\subset [N]$ be a set of frozen sets of size $\rho N$ where $\rho \in (1 - \varepsilon, 0)$. Then, according to \Cref{prop:randomized-w-assumption-const}, we can encode any $m\in \zo^{(1- \rho - \varepsilon)}$. 
    Note that regardless of $\rho$, the fraction of stuck-at bits, our decoding algorithm always succeeds in decoding the encoded message given an encoded vector and the metadata that was generated by Algorithm~\ref{alg:decode-with-assumption}. 
   
    The rest of the proof analyzes the complexity and the metadata size that is transferred to the decoder using an errorless transmission. Clearly, ~\ref{enc1:step1} of the encoding algorithm takes $\cO{N}$ as we just scan the input cover vector and identify the sets of frozen components in each block. The time complexity of Step~\ref{enc1:step2} is the time that it takes to compute a value of the function $g$ which is $\cO{N\cdot \poly{\log N}}$.
    In Step~\ref{enc1:step3}, we solve $M = \cO{\varepsilon \cdot N/\log N}$ linear equation systems, each contains at most $B = \cO{\varepsilon^{-1} \log N}$ equations. 
    Thus, this step can be performed in $\cO{M\cdot B^3} = \cO{\varepsilon^{-2} \cdot N \cdot \log^2 N}$ time.
    Therefore, overall, the encoding algorithm, Algorithm~\ref{alg:encode-with-assumption}, runs in at most $\cO{\varepsilon^{-2} \cdot N \cdot \poly{\log N}}$ time.

    The decoding algorithm reads the value $\bfu_t$ and the value $m_1$ and by applying $g$ on $\bfu_t$, it retrieves the matrices $A_1, \ldots, A_M$. 
    This step takes the time of computing $g$, i.e., $\cO{N \cdot \poly{\log N}}$. Then, recovering each portion of the message $\bfm_i$ is done in $\cO{B^2}$ (simple multiplication of a vector of length $B$ with a matrix of both dimensions $\leq B$). Thus, recovering the entire message $\bfm$ is done in $\cO{M \cdot B^2} = \cO{\varepsilon^{-1} \cdot N \cdot \log N}$. 
    Thus, the overall running time of the decoder is $\cO{\varepsilon^{-1} \cdot N \cdot \poly{\log N}}$. 

    The number of random bits the algorithm needs is 
    \[
    \cO{\log\left( \frac{B \log(B \cdot N)}{N^{-C}}\right)} = \cO{C \log \left( N \right)}\;,
    \]
    and since $\varepsilon = 3/C$, the total number of random bits is $\cO{\varepsilon^{-1} \log N}$.
    Note that metadata, $\bfu$, that is generated by Algorithm~\ref{alg:encode-with-assumption} is of length 
    \begin{equation} \label{eq:metadata-length}
    \cO{C \log N} + \log N + \log B = \cO{C\log N} = \cO{\varepsilon^{-1} \log N} \;.
    \end{equation}
\end{proof}

\begin{remark} \label{rem:deterministic-proof}
    Note that to prove \Cref{thm:explicit-w-assumption-const}, one simply needs to change Step~\ref{enc1:step2} from sampling to a brute force search. Namely, for each one of the $N^{\cO{C}}$ vectors in the space $\zo^t$ (recall that $t=\cO{C\log N}$) we will compute the function $g$ and check if all the matrices $(A_i)_{\overline{F_i}}, i\in [M]$ are full rank. 
    Doing this step now takes $N^{\cO{1/\varepsilon}} \poly{\log N}$ time. Clearly, the complexity of this step dominates the complexity of the algorithm. The rest of proof is identical to that of \Cref{thm:explicit-w-assumption-const}.
\end{remark}

\section{Final construction}

Note that the main issue with the previous construction is that it assumes that we can transmit the decoder the metadata $\bfu$ that contains the data required by the decoder to perform the decoding. In this section, we shall overcome this problem. Intuitively speaking, our solution will encode the message but also the metadata and location of the metadata in our cover object.  
In this section, we prove \Cref{thm:randomized-const} which is restated next.
\randomizedConst*

\subsection{Auxiliary claims}
We start by proving two auxiliary claims that will be useful.
We shall divide our cover vector into four subvectors, $\bfv = \bfv_1 \circ \bfv_2 \circ \bfv_3 \circ \bfv_4$. In $\bfv_1$ and $\bfv_3$ we will encode our message and in $\bfv_2$ and $\bfv_4$ we will make sure that we have enough unfrozen bits to encode the metadata. 
The following claim makes sure that such a partition is indeed feasible. 
Specifically, we will show that for any small enough $\delta$, there exists an interval $[i\cdot \delta N, (i+1)\delta N - 1]$ (we will define $\bfv_2 = \bfv_{[i\cdot \delta N, (i+1)\delta N - 1]}$) such that it contains at most $(\rho + 2\delta) \delta N$ frozen elements and that this interval does not intersect the subvector $\bfv_4$ which contains exactly $N/\log N$ unfrozen bits. 

\begin{claim} \label{clm:encode-meta-ind}
    Let $\rho, \delta \in (0,1)$ such that $2\delta < 1-\rho$ and let $\cF\subseteq [N]$ be of size $\rho N$. Let $j\in [N]$ be the largest such that $\left| [j+1, N]\cap \overline{\cF} \right| = N/\log N$. Then, there exists an integer $i\in [\floor{1/\delta}]$ such that  
    \begin{enumerate}
        \item $i\cdot \delta N + \delta N - 1 \leq j $
        \item $\left| [i \cdot \delta N, (i + 1) \cdot \delta N - 1] \cap F \right| \leq (\rho + 2\delta )\delta N$
    \end{enumerate}
\end{claim}

In order to encode a small number of bits, say $\ell$ bits, one could do the following simple trick. 
Let $x$ represent the decimal number that corresponds to our $\ell$ bits of information, then, one can just flip unfrozen bits such that the weight of the resulting vector is $x (\bmod\,2^{\ell})$. 
The following simple claim shows how many unfrozen bits are needed to encode using this method.
\begin{claim} \label{clm:mod-code}
    Let $\bfv \in \zo^N$ and let $d < N$ be an integer so that there are at least $2d$ unfrozen bits in $\bfv$. Then, we can flip at most $d$ unfrozen bits of $\bfv$ to get a vector $\bfw$ such that $\wt{\bfw} \equiv x (\bmod \, d)$ for any $x < d$.
\end{claim}

We can see that the rate of this encoding method is very small. Indeed, we cannot hope to encode more than $\log N$ bits using this method. 
We will use this method to encode the location of the specific intervals that contain the metadata needed for decoding the message.
\subsection{Encoding and decoding algorithms}

We shall use the following notations throughout this section. Again, $C$ is some universal constant known to the encoder and the decoder.
\begin{itemize}
    \item Let $\delta = 1/C$ and let $B' = \delta N$. Also, assume that $2 \delta < 1 - \rho$.
    \item Let $K$ be the constant implied by \eqref{eq:metadata-length}. Namely, the vector $\bfu$, that is returned from Algorithm~\ref{alg:encode-with-assumption} is of length $K \cdot C\log N$ where $N$ is the length of the cover object.
\end{itemize}

The encoding algorithm is given in Algorithm~\ref{alg:encode} and the decoding algorithm in Algorithm~\ref{alg:decode}. 

Before delving into the details, we give a high-level overview of the encoding algorithm. 
At the first step, we divide $\bfv$ into four contiguous parts, i.e., $\bfv = \bfv_1 \circ \bfv_2 \circ \bfv_3 \circ \bfv_4$ with the premise given in \Cref{clm:encode-meta-ind}. Namely, $|\bfv_2| = B'$ where the number of frozen bits in $\bfv_2$ is at most $(\rho + 2\delta)B'$ and $\bfv_4$ contains exactly $N/ \log N$ unfrozen bits.

The second step invokes Algorithm~\ref{alg:encode-with-assumption} in order to encode the message $\bfm$ into $\bfv_1$ and $\bfv_3$. 
We guarantee that the algorithm will encode only at these parts by adding to the set of frozen bits all the unfrozen bits in $\bfv_2$ and $\bfv_4$. 
Note that this step produces a metadata vector $\bfu_1$ of length at most $KC\log N$ that contains the information needed to decode the message.

The third step first divides $\bfv_2$ into three contiguous parts, $\bfv_2 = \bfv_{21} \circ \bfv_{22} \circ \bfv_{23}$ such that $\bfv_{22}$ is of length $N^{1/2KC}$ and contains at most $(\rho + 2\delta) N^{1/2KC}$ frozen bits. 
Then, we encode $\bfu_1$ in $\bfv_{22}$ using Algorithm~\ref{alg:encode-with-assumption}. This produces a metadata vector $\bfu_2$ of length $KC\cdot \log N^{1/2KC} = \log \sqrt{N}$ which can be represented by a decimal number $U\leq \sqrt{N}$.
We shall encode $\bfu_2$ in $\bfv_{21}$ and $\bfv_{23}$ by flipping at most $2 \sqrt{N}$ bits to make sure that $\wt{\bfv_2}\equiv U (\bmod \, \sqrt{N})$ (see \Cref{clm:mod-code}).

Finally, at the fourth step, we encode the starting position of $\bfv_2$ and $\bfv_{22}$. We will see that both positions can be identified using only $\log (N^{1 - 1/2KC})$ bits. Therefore, by \Cref{clm:mod-code} we can encode this information in $\bfv_4$ by flipping $2 N^{1 - 1/2KC}$ bits and recall that we have $N/\log N$ unfrozen bits there. 

\begin{figure}
    \centering
    \includegraphics[scale=0.5]{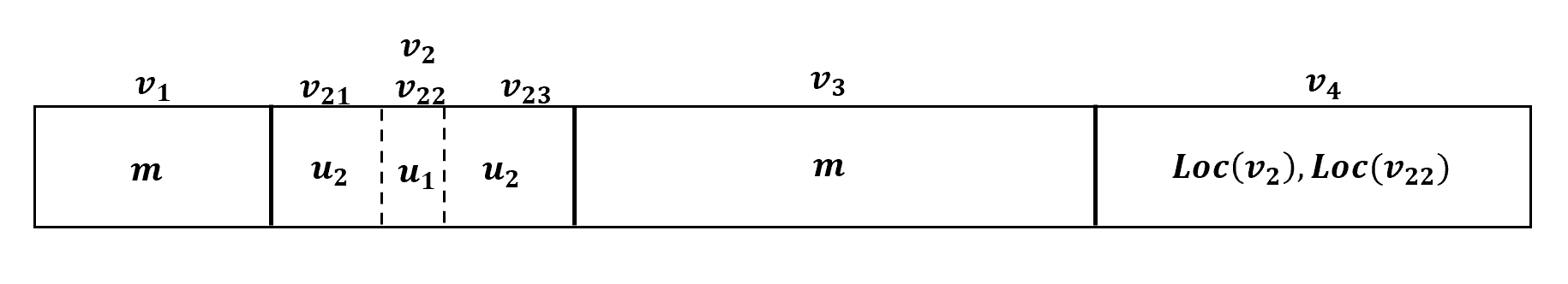}
    \caption{The message $m$ is encoded using Algorithm~\ref{alg:encode-with-assumption} in $\bfv_1$ and $\bfv_3$. Then the metadata $\bfu_1$ that is needed to decode $m$ is encoded in $\bfv_{22}$. 
    The metadata $\bfu_2$ that is needed to decode $\bfu_2$ is encoded using \Cref{clm:mod-code} in $\bfv_{21}$ and $\bfv_{23}$ and the locations of $\bfv_2$ and $\bfv_{22}$ are encoded in $\bfv_4$}
    \label{fig:enter-label}
\end{figure}

\begin{figure}
	\begin{algorithm}[H] 
        \label{alg:encode}
        \caption{Encode}
    	\DontPrintSemicolon
        \LinesNumberedHidden
		\SetNlSty{textbf}{[}{]}
			
		\SetKwInOut{Input}{input}
		\SetKwInOut{Output}{output}
        \SetKwInOut{Notations}{notations}

        \Input{A vector $\bfv \in \zo^N$, a set of frozen indices $\cF\subseteq [N], |\cF| = \rho N$, and a message $\bfm$, of length $ < (1- \rho - \frac{5}{C})N$}
		\Output{A vector $\bfw \in \zo^N$.}
        
        \nlset{1}  \label{enc2:setp1} Find the maximal $j\in [N]$ such that there are at least $N/\log N$ unfrozen coordinates in $\bfv$ to the right of $j$. 
        Find $i\in [j/\delta N]$ such that 
        \[
        \left| [i \cdot B', (i+1)\cdot B' - 1] \cap F \right| \leq \left(\rho +\delta \right) B'
        \]
        and $(i + 1) \cdot \delta N \leq j$.
        Denote $\bfv = \bfv_1\circ \bfv_2 \circ \bfv_3 \circ \bfv_4$ and $\cF = F_1 \cup F_2 \cup F_3 \cup F_4$
        where
        \begin{align*}
            \bfv_1 &= \bfv_{[1,iB' - 1]}  &F_1 &= [1,iB' - 1] \cap F\\
            \bfv_2 &= \bfv_{[i\cdot B', (i+1)\cdot B' - 1]} &F_2 &= [i\cdot B', (i+1)\cdot B' - 1] \cap F  \\
            \bfv_3 &= \bfv_{[(i+1)\cdot B', j-1]} &F_3 &= [(i+1)\cdot B', j-1] \cap F\\
            \bfv_4 &= \bfv_{[j:N]} &F_4 &= [j:N] \cap F
        \end{align*}
        
        \nlset{2} \label{enc2:setp2} Run Algorithm~\Ref{alg:encode-with-assumption} with $\bfv$, $\cF_1\cup [i\cdot B', (i+1)\cdot B' - 1] \cup F_3 \cup [j,N]$, and $\bfm = \bfm$. 
        Denote the first output by $\bfw_1\circ\bfv_2 \circ \bfw_3 \circ \bfv_4$ where $|\bfw_1| = |\bfv_1|$ and $|\bfw_3| = |\bfv_3|$ and second output as $\bfu_1$\;
        
        \nlset{3} \label{enc2:step3} Find $i'\in [i \cdot B', (i+1)\cdot B' - 1]$ such that $i$ is a multiple of $N^{1/2KC}$ and 
        $|[i', i' + N^{1/2KC} - 1] \cap \cF_2| \leq (\rho + 2\delta)N^{1/2KC}$\;
        Denote $\cF_2' = |[i', i' + N^{1/2KC} - 1] \cap \cF_2|$ and $\bfv_2 = \bfv_{21}\circ \bfv_{22} \circ \bfv_{23}$ where
        \begin{align*}
            \bfv_{21} &= \bfv_{[i \cdot B', i' - 1]}\\
            \bfv_{22} &= \bfv_{[i', i' + N^{1/2KC} - 1]}\\
            \bfv_{23} &= \bfv_{[i' + N^{1/2KC}, (i+1) B' - 1]}
        \end{align*}
        \qquad \nlset{3.1} \label{enc2:step3-1} Run Algorithm~\Ref{alg:encode-with-assumption} with $\bfv = \bfv_{22}$, $\cF = F_2'$, and $\bfm = \bfu_1$. Denote the output by $\bfw_{22}$ and $\bfu_2$. Let $U\in [\sqrt{N}]$ be the decimal number that corresponds to $\bfu_2$\;
        \qquad \, \nlset{3.2} \label{enc2:step3-2} Flip unfrozen bits in $\bfv_{21}$ and $\bfv_{23}$ to get $\bfw_{21}$ and $\bfw_{23}$ so that for
        $\bfw_2 := \bfw_{21} \circ \bfw_{22} \circ \bfw_{23}$ it holds that $\wt{\bfw_2} \equiv U (\bmod \, \sqrt{N})$\;
        \nlset{4} \label{enc2:setp4} Let $\bfd = \Bin{i}{1/ \delta} \circ \Bin{i'}{\delta \cdot N^{1 - \frac{1}{2KC}}}$ and denote by $d \in \left[N^{1 - \frac{1}{2KC}} \right]$ the integer whose binary representation is $\bfd$\;
        \,\,\,Flip unfrozen bits in $\bfv_4$ to get the vector $\bfw_4$ where it holds that $\wt{\bfw_1 \circ \bfw_2 \circ \bfw_3 \circ \bfw_4} \equiv d (\bmod \, N^{1 - \frac{1}{2KC}})$ 
    \end{algorithm}
\end{figure}

\begin{figure}
	\begin{algorithm}[H] 
        \label{alg:decode}
        \caption{Decode}
    	\DontPrintSemicolon
        \LinesNumberedHidden
		\SetNlSty{textbf}{[}{]}
			
		\SetKwInOut{Input}{input}
		\SetKwInOut{Output}{output}
        \SetKwInOut{Notations}{notations}

        \Input{A vector $\bfv \in \zo^N$}
		\Output{A message $\bfm\in \zo^*$}

        \nlset{1} \label{dec2:step1} Let $d \equiv \wt{\bfv} (\bmod \, N^{1 - \frac{1}{2KC}})$ and from $\Bin{d}{ N^{\frac{1}{2KC}}}$ identify $i$, the starting of $\bfv_2$, and $i'$, the position inside $\bfv_2$ where $\bfu_1$ is encoded\;
        \nlset{2} \label{dec2:step2} Let $U \equiv \wt{\bfv_{[i, (i+1) \delta N - 1]}} (\bmod \, \sqrt{N})$\;
        \nlset{3} \label{dec2:step3} Run Algorithm~\ref{alg:decode-with-assumption} with input $\bfv_{[i', (i' + 1) N^{1/2KC} - 1]}$ and $\Bin{U}{\sqrt{N}}$ to get $\bfu_1$ \;
        \nlset{4} \label{dec2:step4} Run Algorithm~\ref{alg:decode-with-assumption} with input $\bfv$ and $\bfu_1$ to get $\bfm$\;
    \end{algorithm}
\end{figure}

\subsection{Analyses}
We start by analyzing the rate and then proceed to show the correctness of the algorithms.
\paragraph{Rate.}
Our message $\bfm$ is encoded in Step~\ref{enc2:setp2} by invoking Algorithm~\ref{alg:encode-with-assumption} with $\bfv = \bfv_1 \circ \bfv_2 \circ \bfv_3 \circ \bfv_4$ and a set of frozen elements of size at most $\rho N + \delta N + \frac{N}{\log N}$ (We enlarge the set of frozen elements by adding to $\cF$ all the coordinates of $\bfv_2$ and $\bfv_4$). 
Thus, by the premises of Algorithm~\ref{alg:encode-with-assumption} (see inequality \eqref{eq:const1-rate}), for large enough $N$, we can encode up to
\begin{equation} \label{eq:mess-length-const2}
    (1- \rho -\delta)N - \frac{N}{\log N} - \frac{3N}{C}
\end{equation}
bits which implies that the rate is
\[
 1- \rho - \delta - \frac{3}{C} - \frac{1}{\log N}  \geq 1- \rho - \frac{5}{C}
\]
where the inequality holds for large enough $N$ and by recalling that $\delta = 1/C$.

The correctness is given in the following proposition
\begin{prop} \label{prop:randomized-const}
    Let $\bfv\in \zo^N$, $\cF\subseteq [N]$ where $|\cF| = \rho N$. Let $\bfm \in \zo^{m}$ where $m \leq N(1- \rho - \frac{5}{C} )$ be a message to be encoded. 
    Then, applying Algorithm~\ref{alg:encode} on $\bfv$, $\cF$, and $\bfm$ succeeds with probability at least $1 - \cO{1/\log(N)}$. Furthermore, if Algorithm~\ref{alg:encode} succeeds and outputs the vector $\bfw$ then the decoding algorithm, Algorithm~\ref{alg:decode}, which receives $\bfw$ as input, will output $\bfm$. 
\end{prop}
\begin{proof}
    First note that \Cref{clm:encode-meta-ind} guarantees that the partition that we perform in Step~\ref{enc2:setp1} is indeed possible. 
    
    In Steps~\ref{enc2:setp2}, we invoke Algorithm~\ref{alg:encode-with-assumption}. 
    In doing that, we have to make sure that the input we give the algorithm is valid. Specifically, if one wishes to encode a message $\bfm$ of length $m$ in a vector $\bfv$ of length $\overline{N}$ with a set of frozen indices $\cF$, then by \eqref{eq:const1-rate} we need that
    \begin{equation} \label{eq:mess-len-ineq}
        |\bfm| \leq \overline{N}  - |\cF| - \frac{3\overline{N}}{C}\;.
    \end{equation}
    We already showed in \eqref{eq:mess-length-const2} what is the maximal message length that can be encoded in Step~\ref{enc2:setp2} and that our message length is below that threshold for large enough $N$.

    In Step~\ref{enc2:step3}, we focus just on $\bfv_2 = \bfv_{[i\cdot B', i\cdot  B' - 1]}$. We first find an index $i'\in [i\cdot B', i\cdot  B' - 1]$ that is a multiple of $N^{1/2KC}$ such that $\bfv_{22} = \bfv_{[i', i' + N^{1/2KC} - 1]}$ has at most $(p + \delta)N^{1/2KC}$ frozen bits. 
    Since $\bfv_2$ contains at most $(\rho + \delta)B'$ frozen bits, such an index must exist by a simple averaging argument. 
    Then, in Step~\ref{enc2:step3-1}, our goal is to encode $\bfu_1$, the metadata that was returned at the previous step, in $\bfv_{22}$. 
    By the proof of \Cref{thm:randomized-w-assumption-const}, the length of $\bfu_1$ is $K C\log N$ (recall that $K$ is the constant implied by \eqref{eq:metadata-length}). 
    Now, since $|\bfv_{22}| = N^{1/2KC}$ and $|\cF_2'| \leq (\rho + \delta) N^{1/2KC}$ ($\cF_2'$ is the set of frozen coordinates in $\bfv_22$), then, for large enough $N$, 
    inequality \eqref{eq:mess-len-ineq}, holds with $\bfm = \bfu_1$, $\cF = \cF_2'$, and $\overline{N} = N^{1/2KC}$.
    Recall that Algorithm~\ref{alg:encode-with-assumption} returns the encoded vector, which we call $\bfw_22$, and a metadata vector, $\bfu_2$, that is needed to decode $\bfw_22$.   The length of $\bfu_2$ is $KC\log (N^{1/2KC}) = \log (\sqrt{N})$ and our next goal it using \Cref{clm:encode-meta-ind}.
    
    To encode $\bfu_2$ in Step~\ref{enc2:step3-2}, we represent it using a decimal number $U$ which is at most $\sqrt{N}$. 
    Observe that the number of unfrozen bits in $\bfv_{21}$ and $\bfv_{23}$ is at least $(1-\rho - 2\delta)\delta N - N^{1/2KC}$ which is greater than $2\sqrt{N}$, for large enough $N$. 
    Therefore, by \Cref{clm:mod-code}, we can flip $\sqrt{N}$ unfrozen bits in $\bfv_{21}$ and $\bfv_{23}$ and make sure that the resulting weight of $\bfw_2 = \bfw_{21} \circ \bfw_{22} \circ \bfw_{23}$ will be $U (\bmod \, \sqrt{N})$.
    We now compute what is the failure probability of Steps~\ref{enc2:setp2} and ~\ref{enc2:step3}. Recall that Algorithm~\ref{alg:encode-with-assumption} can fail with probability $\cO{1/ (C\log N)}$, therefore, the failure probabilities of Step~\ref{enc2:setp2} and ~\ref{enc2:step3} are $\cO{1/ (C\log N)}$ and $\cO{2K/\log N}$.

    To convince ourselves that Step~\ref{enc2:setp4} is feasible, we note that the maximal value $i$ can take is upper bounded by $C$ and that the value of $i'$ is upper bounded by $(N/C) \cdot N^{1/2KC} = N^{1 - 1/2KC}/C$. 
    Therefore, $|\Bin{i}{C} \circ \Bin{i'}{N^{1 - 1/2KC}/C}| = \log{N^{1 - 1/2KC}}$ which implies that there exists a decimal number $U' \leq N^{1 - 1/2KC}$ that corresponds uniquely to the values $i$ and $i'$. 
    Note that as the number of unfrozen bits in $\bfv_4$ is $N/\log N > 2N^{1 - 1/2KC}$ (where the inequality is for large enough $N$), by \Cref{clm:mod-code}, we can flip these unfrozen bits in $\bfv_4$ to make sure that the weight of $\bfw$ is $U' (\bmod \,N^{1 - 1/2KC})$.

    As for the decoder. Note that by computing the weight of $\bfv$ (in Step~\ref{dec2:step1}), the decoder knows the values of $i$ and $i'$. 
    Thus, he knows that he needs to compute $\wt{v_{[i, i + B' - 1]}}$ in order to get the metadata that is needed for decoding $\bfu_1$ from $\bfv_{[i', i' + N^{1/2KC} - 1]}$. 
    Once he extracts $\bfu_1$ from $\bfv_{[i', i' + N^{1/2KC} - 1]}$ in Step~\ref{dec2:step3}, he can proceed to Step~\ref{dec2:step4} and decode the message.
\end{proof}

\begin{proof}[Proof of \Cref{thm:randomized-const}]
    Let $C$ be a universal constant and define $\varepsilon = 5/C$. Note here that $\delta = 1/C < (1-\rho)/5$ and thus our assumption that $2\delta < 1-\rho$ holds. Let $N$ be a large enough integer (depending only on $\varepsilon$) and let $\bfv\in \zo^{N}$ be a cover vector, and $\cF\subset [N]$ be a set of frozen sets of size $\rho N$ where $\rho \in (1 - \varepsilon, 0)$. Then, according to \Cref{prop:randomized-const}, we can encode any $m\in \zo^{(1- \rho - \varepsilon)}$. 
    Note that regardless of $\rho$, the fraction of stuck-at bits, our decoding algorithm always succeeds in decoding the encoded message given an encoded vector Algorithm~\ref{alg:encode}. 
    
    We are left to show that the complexity is $N\cdot \poly{N} \cdot \poly{1/\varepsilon}$ for both the encoder and the decoder. 
    Clearly, Step~\ref{enc2:setp1} can be done in $\cO{N}$. Indeed, identifying a subvector of a specific length with a maximal number of unfrozen bits requires a single scan of the entire input vector.
    Steps~\ref{enc2:setp2} and ~\ref{enc2:step3} both invoke Algorithm~\ref{alg:encode-with-assumption}, and thus their running time is $N\cdot \poly{N} \cdot \poly{1/\varepsilon}$. 
    Steps~\ref{enc2:setp4} requires at most $\cO{N}$ as we need just to compute the weight of a vector and then flip at most $N^{1 - 1/2KC}$ bits. Note that finding the bits that need to be flipped takes also $\cO{N}$ time since we need to find the dominant symbol in the unfrozen bits ($0$ or $1$) and then flip the first unfrozen occurrences of that symbol. Thus, the encoder runs in $N\cdot \poly{N} \cdot \poly{1/\varepsilon}$ time, as desired.

    We analyze now the decoder. Steps~\ref{dec2:step1} and~\ref{dec2:step2} run in time $\cO{N}$ as we compute the weight of a vector and perform a casting of an integer in decimal representation to binary representation. Steps~\ref{dec2:step3} and~\ref{dec2:step4} invoke Algorithm~\ref{alg:decode-with-assumption} whose running time is $N\cdot \poly{N} \cdot \poly{1/\varepsilon}$. Thus, the total running time is again $N\cdot \poly{N} \cdot \poly{1/\varepsilon}$.
\end{proof}

\begin{remark}
    As discussed in \Cref{rem:deterministic-proof},
    to prove \Cref{thm:deterministic-const}, which is the deterministic version of \Cref{thm:randomized-const}, we change the second step in Algorithm~\ref{alg:encode-with-assumption} to a brute force step. This change affects the complexity of Step~\ref{enc2:setp2} and~\ref{enc2:step3} which now becomes $N^{\cO{1/\varepsilon}}$.
    The rest of the proof is identical to the one of the randomized construction.
\end{remark}
\section{Acknowledgements}
The first author would like to thank Dean Doron and Jo\~ao Ribeiro for helpful discussions about \cite{alon1992simple}. 
\bibliographystyle{alpha}
\bibliography{refs}

\newpage
\appendix
\section{Appendix}
\subsection{The complexity of $g$ from \Cref{lem:k-wise}}
We briefly recall the third construction given in \cite{alon1992simple}. 

\begin{itemize}
    \item Let $h := k\log r$ and let $A\in \mathbb{F}_2^{r \times h}$ be a generating of binary code whose dual distance is exactly $k$.
    \item Let $t := \log \frac{h}{\mu}$
    \item Let $x,y\in \mathbb{F}_{2^{t/2}}$, where $\mathbb{F}_{2^{t/2}}$ is the finite field with $2^{t/2}$ elements. 
            Note that $x$ and $y$ can be viewed also as elements in $\zo^{t/2}$ as $\mathbb{F}_{2^{t/2}} \cong \mathbb{F}_2 ^{t/2}$.
\end{itemize}

The function $g:\zo^{t} \rightarrow \zo^r$ is defined by 
\[
g(x,y) = A \cdot \left( \langle x^0,y \rangle, \langle x^1, y \rangle, \ldots, \langle x^{h-1}, y \rangle \right)
\]

where $\langle \cdot, \cdot \rangle$ is the mod two inner product. By \cite{alon1992simple}, $g(U_t)$ is $\mu$-wise $k$ independent  variable over $\zo^r$ where $U_t$ is the uniform distribution over $U_t$.

As for the complexity of computing $g$. Note that $\langle x^i, y \rangle$ can be performed in $\poly{t}$ time and we perform this operation $h$ times. The multiplication of the vactor by the matrix takes $\cO{h \cdot r}$ operations.
In total, we perform, 
\[
\cO{h\cdot \poly{t} + h\cdot r} = \cO{k \log r \cdot \poly{\log \frac{k \log r}{\mu}} + r \cdot k\log r}\;,
\]
and since in our settings, $r = \cO{N\log N}$, $k = \cO{\log N}$, and $\mu = N^{\cO{1}}$, we get that the complexity of $g$ is 
\[
\cO{\poly{\log N} + N\cdot \poly{\log N}} = \cO{N\cdot \poly{\log N}}\;.
\]

\subsection{Missing proofs}
\begin{proof}[Proof of \Cref{clm:full-rank-pro}]
    Denote by $r_i$ the $i$th row. The matrix $A$ has full rank if and only if for any $i\in [m]$, 
    \begin{itemize}
        \item The rows $r_1, \ldots, r_{i-1}$ are linearly independent and,
        \item $r_i \notin \text{span} \{r_1,\ldots, r_{i-1} \}$.
    \end{itemize}
    Thus, 
    \begin{align*}
        \Pr[A \text{ has full rank}] &= \prod_{i=1}^{m} \Pr[r_i \notin \text{span} \{r_1,\ldots, r_{i-1}\}] \\
            &\geq \prod_{i=1}^{m} \left( 1 - 2^{i-1}(2^{-n} + \varepsilon)\right) \\
            &\geq 1 - (2^{-n} + \varepsilon) \cdot \sum_{i=1}^m 2^{i-1} \\
            & = 1 - 2^{m-n} - \varepsilon \cdot 2^m 
    \end{align*}
    where the first inequality holds since $r_i=(A_{r,1},A_{r,2}, \ldots, A_{r,n})$ is an $\varepsilon$-almost $n$-wise random variable and the second inequality is a standard union bound.
\end{proof}

\begin{proof}[Proof of \Cref{clm:encode-meta-ind}]
    Denote $x = |[j+1, N]\cap \cF|$ (the number of frozen bits in $\bfv_4$) and note that $j = N - x - \frac{N}{\log N}$. Therefore, we have $|[j] \cap \cF| = \rho N - x$. Denote $I_{i} := [i\cdot \delta N, (i+1) \cdot \delta N - 1]$ for all $i\in [\floor{j/\delta N} - 1]$. 
    Assume that for all $i\in [\floor{j/\delta N} - 1]$, it holds that $|I_{i} \cap F| > (\rho + 2\delta) \delta N$. 
    Then, $|[\delta N \cdot \floor{j/\delta N }] \cap F| > (\rho + 2\delta ) j - \delta N$.
    Therefore,
    \begin{align*}
        |\cF| &>  \left(\rho + 2\delta \right) j - \delta N + x \\
            &= (\rho + 2\delta) N + (1 - \rho - 2\delta) x - \frac{(\rho + 2\delta)N}{\log N} - \delta N\\
            &= (\rho + \delta) N + (1 - \rho - 2\delta)x -\cO{\frac{N}{\log N}}\\
            & > \rho N
    \end{align*}
    where the last inequality follows since $2\delta < 1-\rho$ and for large enough $N$.
\end{proof}

\begin{proof}[Proof of \Cref{clm:mod-code}]
    Assume that $\wt{\bfv}\equiv y (\bmod \, d)$ for some $y < d$. Since there are at least $2d$ unfrozen bits, at least $d$ of them are either $1$ or $0$. Assume w.l.o.g., that at least $d$ of them are zero. Then, we need to flip exactly $x - y (\bmod \, d)$ bits in order to get a vector $\bfw$ with $\wt{\bfw} \equiv x (\bmod \, d)$.
\end{proof}
\end{document}